\documentclass[runningheads]{llncs}
\usepackage[T1]{fontenc}
\usepackage{tikz}
\usepackage{amsmath}
\usepackage{amssymb}
\usepackage{algorithm,algpseudocode}
\usepackage{bbding}
\usepackage{pgfplots}
\usepackage{subcaption}
\usepackage{array}
\usepackage[colorlinks=true, citecolor=blue,linkcolor=blue, urlcolor=blue]{hyperref}
\usepackage{graphicx}
\begin{document}
\title{Imitater: An Efficient Shared Mempool Protocol \\ with Application to Byzantine Fault Tolerance}
%
\titlerunning{Imitater: An Efficient Shared Mempool Protocol}
%
\author{Qingming Zeng\inst{1} \and
Mo Li\inst{2} \and
Ximing Fu\thanks{Corresponding author}\inst{1,5,6}
\and Hui Jiang\inst{3,4}
\and Chuanyi Liu\inst{1,5,6}
}

\authorrunning{Q. Zeng et al.}


\institute{Harbin Institute of Technology, Shenzhen, China \\ \email{22S151098@stu.hit.edu.cn, \{fuximing,liuchuanyi\}@hit.edu.cn} \and
The Chinese University of Hongkong, Shenzhen, China
\email{220019160@link.cuhk.edu.cn}\\
\and Tsinghua University, Beijing, China
\and Baidu Inc, Beijing, China \\
\email{jianghui01@baidu.com}
\and Peng Cheng Laboratory, Shenzhen, China \\
\and Key Laboratory of Cyberspace and Data Security, Ministry of Emergency Management
}

\maketitle              
\begin{abstract}

Byzantine Fault Tolerant (BFT) consensus, a cornerstone of blockchain technology, has seen significant advancements. While existing BFT protocols ensure security guarantees, they often suffer from efficiency challenges, particularly under conditions of network instability or malicious exploitation of system mechanisms.

We propose a novel Shared Mempool (SMP) protocol, named Imitater, which can be seamlessly integrated into BFT protocols. 
By chaining microblocks and applying coding techniques,
Imitater efficiently achieves \emph{totality} and \emph{availability}. Furthermore, a BFT protocol augmented with Imitater ensures \emph{order preservation} of client transactions while mitigating the risks of \emph{over-distribution} and \emph{unbalanced workload}.

In the experiment, we integrate Imitater into the HotStuff protocol, resulting in Imitater-HS. The performance of Imitater-HS is validated in a system with up to 256 nodes. Experimental results demonstrate the efficiency of our approach: Imitater-HS achieves higher throughput and lower latency in the presence of faulty nodes compared to Stratus-HS, the state-of-the-art protocol. Notably, the throughput improvement increases with the number of faulty nodes.

\end{abstract}
\section{Introduction}


With the rapid advancement of blockchain technology, Byzantine Fault Tolerant (BFT) consensus has become a cornerstone of modern blockchain systems, driving extensive research and applications~\cite{aptos,diem,ethereum,androulaki2018hyperledger}.
At its core, BFT consensus protocols are designed to implement State Machine Replication (SMR) in the presence of Byzantine faults, ensuring consistent transaction ordering across distributed nodes by maintaining identical state machines and logs.
They provide strong guarantees of both \textit{safety}—all honest nodes agree on a single transaction order—and \textit{liveness}—all valid client requests are eventually processed despite Byzantine faults.



Existing methods focus on optimizing the message complexity of protocols to enhance performance \cite{castro1999practical,kotla2007zyzzyva,jalalzai2023fast}, culminating in HotStuff \cite{yin2019hotstuff}, which achieves linear complexity. 
Most of these protocols are leader-based, where a designated principal node packages transactions to propose, and other replicas vote on the leader's proposal to reach consensus.
Under the framework of the leader-replica structure, the leader becomes a bottleneck, thereby affecting scalability. 
To address this, some approaches adopt scale-out via sharding~\cite{kokoris2018omniledger,wang2019monoxide,zamani2018rapidchain}, dividing replicas into groups.
However, these methods rely on a different fault model that limits fault tolerance within each group, rather than across the entire system.

\subsubsection{Shared Mempool Based Protocols.}
The scalability bottleneck in leader-based BFT protocols primarily stems from the leader's responsibility for transaction distribution. 
Recent studies~\cite{hu2022leopard,danezis2022narwhal,gai2023scaling,hu2023data} have explored another approach to improving the scalability of BFT consensus by decoupling the block distribution and consensus phases, as described in \cite{biely2012s,zhao2018sdpaxos,bagaria2019prism}, and have demonstrated significant performance improvements in recent implementations~\cite{danezis2022narwhal}.
Each node independently packages its local transactions into microblocks and distributes them to other nodes. The leader then aggregates the identifiers (e.g., hash values) of these microblocks into a candidate block for consensus at the start of each consensus round. This decoupling technique improves the utilization of both computational and bandwidth resources. This decoupling approach also abstracts the microblock distribution phase as a Shared Mempool (SMP)~\cite{gai2023scaling,hu2023data}.

Under the decoupled framework, a key challenge is ensuring the microblock availability property. When a replica receives a proposal, the microblocks it references may not be locally available, 
requiring it to wait until all are available before proceeding (e.g., voting).
This may involve passive waiting or active requests to other nodes—both of which can significantly delay consensus.

To address this, existing Shared Mempool implementations~\cite{danezis2022narwhal,gai2023scaling} ensure microblock availability by having replicas to generate an availability certificate when distributing microblocks. The leader attaches these certificates to the proposal, enabling replicas to proceed with subsequent steps as soon as they receive the proposal. The certificates assure replicas that all referenced microblocks will eventually be received, eliminating the need for local availability at the time of proposal processing.




\subsubsection{Efficient Totality.}
An essential property of Byzantine reliable broadcast~\cite{cachin2011introduction} is \textbf{totality}, which ensures that if an honest node delivers a message $v$, then all honest nodes will eventually deliver $v$. This property is equally critical in BFT state machine replication (BFT-SMR). Specifically:

\begin{itemize}
    \item[$\cdot$] \textbf{Sequential Execution.} In State Machine Replication, the execution of transactions is contingent upon the order established by consensus. However, Byzantine faults or the asynchronous nature of the network may cause some nodes to miss certain blocks. When blocks are missing, subsequent blocks, despite achieving consensus, cannot be executed due to the sequential execution requirement.
    
    \item[$\cdot$] \textbf{Garbage Collection.} From the implementation perspective, nodes aim to release memory resources associated with a given request as soon as a decision is reached. However, because other nodes may later request this data, nodes cannot predict when it is safe to release the associated memory, forcing them to maintain these resources indefinitely.    
\end{itemize}

It is worth noting that achieving totality is not inherently guaranteed by microblock availability. Most current systems do not address totality explicitly within BFT protocols. Instead, they rely on alternative mechanisms, such as state transfer (i.e., reliably broadcasting the
decision value of each block)\cite{castro1999practical,kotla2007zyzzyva} or a pull mechanism (i.e., requesting missing data from other nodes)\cite{danezis2022narwhal,hu2023data}. While resource consumption prior to reaching consensus is often a primary focus, post-consensus resource usage also deserves significant attention. Both state transfer and pull mechanisms operate after the consensus decision point, potentially contributing to further resource overhead.


However, these approaches are often communication-inefficient and may even undermine the original performance of the protocol. State transfer requires quadratic communication overhead, which remains unavoidable even in synchronous environments~\cite{dolev1983authenticated}. For the pull mechanism, considering the presence of faulty nodes, they might request every block from all nodes, regardless of whether the block is genuinely missing, thus imposing unnecessary bandwidth consumption on the system. In the worst-case scenario, if all faulty nodes request blocks from all honest nodes, this could result in quadratic bandwidth consumption, severely degrading system performance. Even one faulty node could halve the throughput, which will be confirmed in the experiments in Section~\ref{Sec:evaluation}



\subsubsection{Contribution.} In this paper, we pay efforts for solving the above concern. To encounter the above concerns, a series of techniques are exploited, making the following contributions.

We propose a novel Shared Mempool (SMP) protocol (Section~\ref{sec:imitater}) named Imitater that efficiently guarantees the \emph{totality} of microblock distribution through the use of erasure codes, ensuring that all microblocks are eventually ready for execution without relying on communication-intensive request-based mechanisms.

We demonstrate how to compile Imitater into an efficient BFT protocol, Imitater-BFT(Section~\ref{sec:building-bft}), which offers the following advantages: (i) the \emph{totality} and \emph{microblocks availability} can be naturally inherited from the SMP totality; (ii) by appropriately scheduling the Dispersal and Retrieval components within SMP, we achieve \emph{bandwidth adaptability} (Section~\ref{sec:trigger-dispersal});  (iii) the protocol ensures \emph{order keeping} of client transactions in the final decision, effectively handles \emph{over distribution} risks, and accommodates \emph{unbalanced workload} scenarios without compromising protocol performance (Section~\ref{sec:practical-problems}).

We integrate Imitater into the HotStuff protocol, creating Imitater-HS, and validate it through large-scale experiments(Section~\ref{Sec:evaluation}). 
In a 100-node system with up to 1/3 faulty nodes, Imitater-HS achieves 21.5Kops/s throughput—approximately 9× higher than Stratus-HS's 2.3Kops/s—demonstrating its superior performance and practical scalability.

\section{Background}
\label{Section:background}
\subsection{System Model}

In our setup, the system consists of \(n = 3f + 1\) \textit{nodes} or \textit{replicas}, with up to \(f\) Byzantine nodes controlled by an adversary \(\mathcal{A}\). Honest nodes strictly follow the protocol, while Byzantine nodes can behave arbitrarily. We assume reliable communication and a public key infrastructure (PKI). The adversary \(\mathcal{A}\) has limited computational power and cannot break cryptographic primitives.

We consider a partially synchronous network as outlined in \cite{dwork1988consensus}. After an unknown Global Stabilization Time (GST), messages transmitted between honest nodes arrive within a known maximum bound \(\Delta\). However, the transmission of network messages can be manipulated by the adversary \(\mathcal{A}\), who may delay or reorder messages. Our consensus protocol follows the \emph{leader-replica} framework, with the leader proposing blocks and the replicas deciding to accept or deny.

External clients continuously submit transaction requests to the system, choosing replicas randomly, by proximity, or preferentially. Given that Byzantine nodes may censor transactions, clients use a timeout mechanism to resubmit transactions to other replicas until they find an honest node.

\subsection{Primitives}

We use a \((2f+1,n)\) threshold signature scheme~\cite{boneh2001short,shoup2000practical},
which includes algorithms \((TSign, TComb, TVerf)\). The \(TSign\) algorithm allows node \(i\) to sign message \(m\) with its private key \(tsk_i\), producing a partial signature \(\sigma_i\). The \(TComb\) algorithm combines \(2f+1\) partial signatures for the same message \(m\) into an aggregated signature \(\sigma\). The \(TVerf\) algorithm then verifies \(\sigma\) using the public key \(pk\), message \(m\), and the combined signature \(\sigma\).

We employ an \((f+1,n)\)-erasure code, Reed-Solomon code \cite{reed1960polynomial} by default, which includes a pair of encoding and decoding algorithms $( Enc, Dec )$. The encoding algorithm $Enc$ encodes a message into \(n\) chunks $s_1,s_2,...,s_n$, while the decoding algorithm decodes the message using any \(f+1\) out of $n$ coded chunks.

We utilize Merkle trees~\cite{kocher1998certificate} to demonstrate the set relationships of coded chunks, including a set of Merkle tree generation functions and Merkle proof verification functions denoted as \((M.Gen, M.Vrf)\). We use \(R(\cdot)\) to represent the root of the Merkle tree and \(\boldsymbol{P}(\cdot)\) to denote the Merkle proofs. \(M.Gen\) generates the Merkle root and the Merkle proofs for a set of elements $\boldsymbol{s}$, i.e., \((\boldsymbol{P}(\boldsymbol{s}), R(\boldsymbol{s})) \gets M.Gen(\boldsymbol{s})\), while \(M.Vrf\) verifies the validity of a proof $P_i(\boldsymbol{s})$ for an element $s_i$, i.e., \(\{True, False\} \gets M.Vrf(P_i(\boldsymbol{s}), s_i, R(\boldsymbol{s}))\).

\section{Imitater Protocol}
\label{sec:imitater}
\subsection{Overview}


In this paper, we adopt a Shared Mempool (SMP)-based architecture~\cite{hu2022leopard,danezis2022narwhal,gai2023scaling,hu2023data}, where each node independently packages its local transactions into microblocks and distributes them to other nodes. The leader then packages the identifiers (e.g., hash values) of these microblocks into a candidate block for consensus. The distributed microblocks is referred to as the Shared Mempool (SMP)~\cite{gai2023scaling,hu2023data}, which conceptually represents a shared memory pool. We use the term \textit{mempool} to denote each node’s actual local storage.

\begin{figure}
    \centering
    \includegraphics[width=0.5\textwidth]{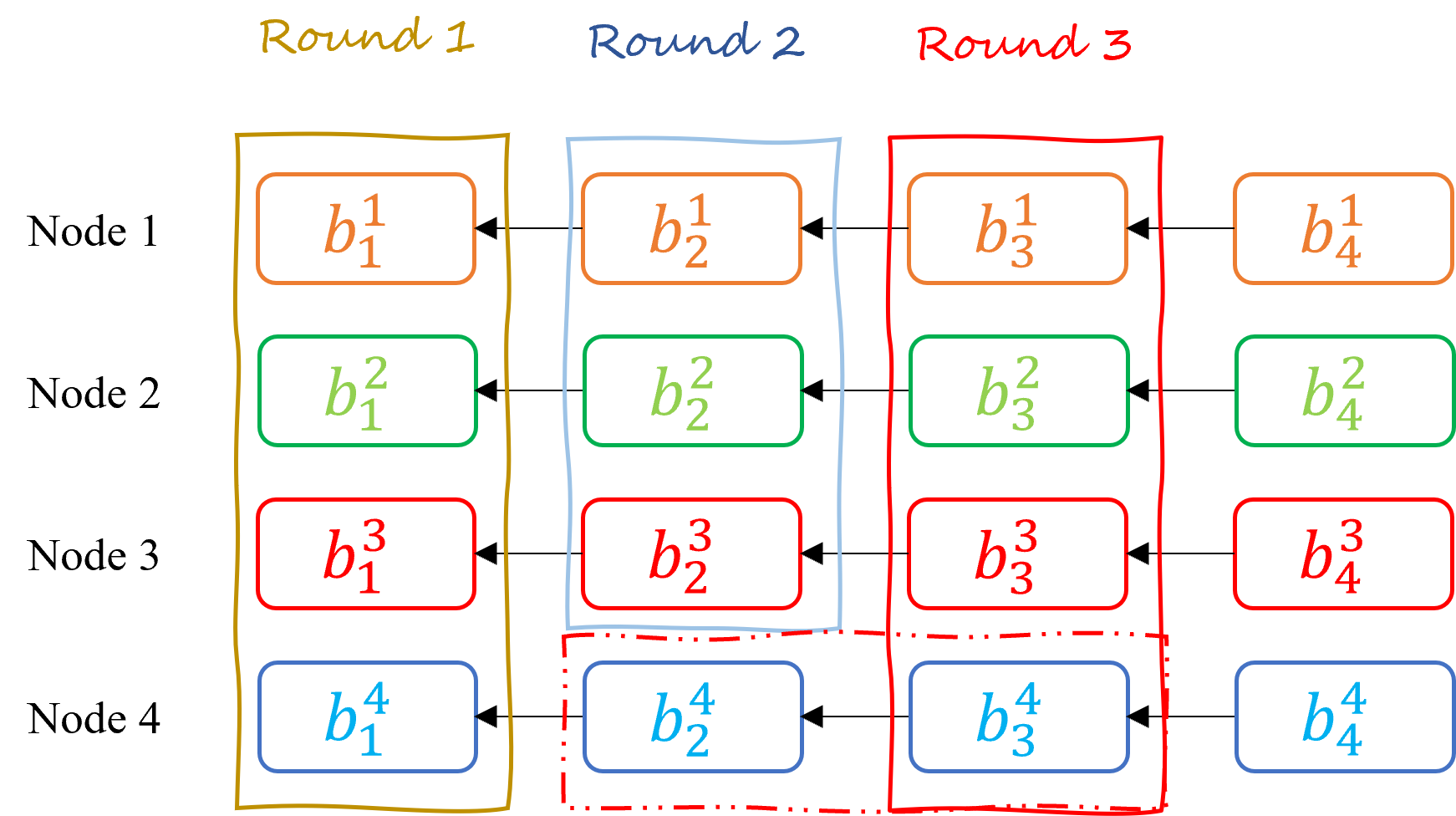}
    \caption{Chain-based organization of a node's microblock mempool. Each node maintains $n$ local microblock chains, where the $i$-th chain sequentially records microblocks disseminated by node $i$.}

    \label{fig:Imitater}
\end{figure}
\subsubsection{Chained Microblocks.}
Imitater is designed in the form of a chained structure, as shown in Figure~\ref{fig:Imitater}. 
Let \(b_p^{i,j}\) denote a microblock at position \(p\) on the \(i\)th chain in the mempool of node \(j\) and $R(b_p^{i,j})$ be the identifier of \(b_p^{i,j}\). In this paper, we use the Merkle root $R(\cdot)$ as the identifier of an microblock, which we will elaborate in Section~\ref{sec:dis-phase}. Since all honest nodes eventually store the same chains, as shown in Lemma~\ref{lemma:Microblock consistency}, \(j\) is omitted, and the microblock is denoted as \(b_p^i\) for simplicity.

When node $i$ distributes a microblock, each node receiving the distribution message replies an ack message together with a partial signature of the corresponding identifier. 
The partial signatures are aggregated into a threshold signature \( \sigma(b_p^i) \), forming the \emph{Availability Certificate (AC)} \(C_p^i = (\sigma(b_p^i), R(b_p^i))\). 
Once node $i$ has collected the AC for \(b_p^i\), it can construct the next microblock \(b_{p+1}^i\) and distribute it, hence \(b_{p+1}^i\) implies the availability of \(b_{p}^i\). Since a microblock includes the AC of its predecessor, this naturally forms a chain.



Now we formally define the structure of a microblock as \( b_p^i = (R(b_p^i), C_{p-1}^i, t_p^i) \), consisting of the identifier $R(b_p^i)$, the availability certificate of the previous block $C_{p-1}^i$, and the content $t_p^i$ of the current block.


A chain-based SMP efficiently maintains the order of microblocks and reduces proposal size during consensus. At the start of each round, the leader packages the highest-position identifiers from each chain into the proposal. If some identifiers are not included, later proposals include the latest identifier, implicitly incorporating all preceding uncommitted microblocks of the chain, ensuring efficient consensus and sequential integrity.


For example, in Figure~\ref{fig:Imitater}, \(b_2^4\) and \(b_3^4\) were not included in earlier proposals. When a successor microblock such as \(b_4^4\) is included, \(b_2^4\) and \(b_3^4\) are implicitly added in the correct order, preserving their sequence.

\subsubsection{Coding and Totality.}
To avoid the request-based methods with high communication overhead, Imitater addresses totality by exploiting coding techniques. More concretely, our approach to microblock distribution involves distributing encoded chunks, nodes forwarding and collecting these chunks, instead of directly broadcasting the complete microblocks. 

Distributing blocks via coding techniques has been extensively explored in previous works~\cite{miller2016honey,kaklamanis2022poster,yang2022dispersedledger}. 
Imitater repurposes coding technique to solve the totality problem, ensuring complete microblock delivery to all nodes.
In addition, this method of coding produces \emph{bandwidth adaptability}, detailed in Section~\ref{sec:trigger-dispersal}.

\subsection{SMP Protocol}
\label{sec:smp detail}
Our SMP protocol, based on the communication model in Figure~\ref{fig:Imitater-Comm}, consists of two phases: \textbf{Dispersal} and \textbf{Retrieval}, defined in Algorithms~\ref{Alg:Microblock Dis} and \ref{Alg:Microblock Ret}. The Dispersal phase handles chunk distribution and Availability Certificate (AC) generation, while the Retrieval phase ensures all honest nodes acquire complete microblock content.

\label{sec:SMP_detail}

\begin{figure}[t]
    \centering
    \includegraphics[width=0.6\textwidth, height=3.5cm]{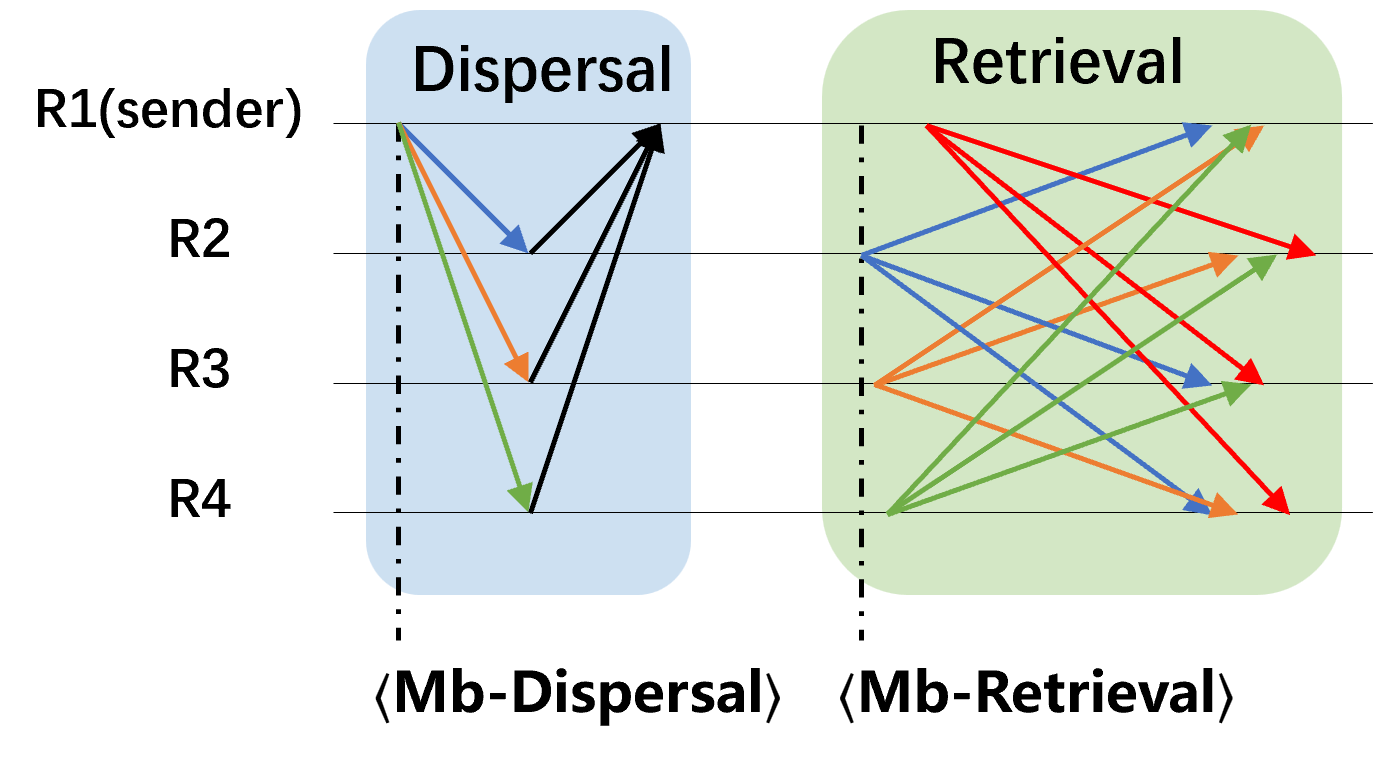}
    \caption{The communication pattern of disseminating a microblock in SMP}
    \label{fig:Imitater-Comm}
\end{figure}


\subsubsection{Dispersal Phase.}
\label{sec:dis-phase}
The Dispersal phase is initiated by an Mb-Dispersal event. When a node $i$ distributes a microblock \(b_p^i\), it must utilize the AC for its predecessor \(b_{p-1}^i\). Upon triggering an Mb-Dispersal event, node $i$ employs $Enc$ to encode the microblock \(b_p^i\) into \(n\) chunks $s^{i,1}_p,\ldots,s^{i,n}_p$. It then uses $M.Gen$ to compute the corresponding Merkle root \(R(b_p^i)\) and Merkle proof \(\boldsymbol{P}(s_p^{i,j})\) for each $j\in [n]$. The Merkle root enables direct verification of the correspondence between any given data chunk and its parent microblock, thus serving as its identifier. Node \(i\) then sends the message \(\langle\text{Mb-Dis},R(b_p^i),C_{p-1}^i,s^{i,j}_p,P(s_p^{i,j})\rangle_{i}\) to all nodes $j$. 

Upon receiving an Mb-Dis message, node $j$ verifies the AC \( C_{p-1}^i \) and the Merkle proof using $TVrf$ and $M.Vrf$. It also ensures that this is the first Mb-Dis message for position $p$ from node $i$; otherwise, the message is ignored. 
After passing all verification, node \(j\) stores the data in its local SMP and signs \(R(b_p^i)\) using \(TSign\), sending the resulting signature \(\sigma_j(b_p^i)\) in an Mb-Ack message to node \(i\).
Once node $i$ collects \(2f+1\) Mb-Ack messages for \(R(b_p^i)\) from different nodes, it aggregates the signatures into a proof \(\sigma(b_p^i)\) using \(TComb\). Thus, \((\sigma(b_p^i),R(b_p^i))\) forms the AC for \(b_p^i\). Therefore, node $i$ can utilize this certificate to initiate the Dispersal of its next microblock \(b_{p+1}^i\).


\begin{algorithm}[t]
    \caption{Microblock Dispersal with $b_p^i$ at node $i$}
    \label{Alg:Microblock Dis}
    \begin{algorithmic}[1]
        \item Local variables:
        \item $\boldsymbol{Sigs}\gets \{\}$  \textcolor{blue}{\Comment{Signatures over $R(b_p^i)$}}
        
        \vspace{0.5em}
        \item \textbf{Upon} event $\langle\text{Mb-Dispersal},b_p^i,C_{p-1}^i\rangle$ \textbf{do}
        \item \hspace{1em}$s_p^{i,1},...,s_p^{i,n}\gets Enc(b_p^i)$ \textcolor{blue}{\Comment{Encode $b_p^i$ into $n$ chunks}}
        \item \hspace{1em}$\boldsymbol{P}(b_p^i)$, $R(b_p^i)\gets M.Gen(s_p^{i,1}, \dots, s_p^{i,n})$
        \item \hspace{1em}\textbf{for} $j\in [n]$ \textbf{do}
        
        \item \hspace{2em}Send $\langle\text{Mb-Dis},R(b_p^i),C_{p-1}^i,s_p^{i,j},P_j(b_p^i)\rangle_{i}$ to $j$. 
        
        \vspace{0.5em}
        \item \textbf{Upon} receiving $\langle\text{Mb-Dis},R(b_p^i),C_{p-1}^i,s_p^{i,j},P_j(b_p^i)\rangle_i$ for the first time \textbf{do}
        \item \hspace{1em} \textbf{If} $TVerf(C_{p-1}^i)$ and $M.Vrf(P_j(b_p^i),s_p^{i,j},R(b_p^i))$ \textbf{do}
        \item \hspace{2em}Store $(R(b_p^i),s_p^{i,j})$
        \item \hspace{2em}$\sigma_j(b_p^i)\gets TSign(R(b_p^i),tsk_j)$
        \item \hspace{2em}Send $\langle\text{Mb-Ack},R(b_p^i),\sigma_j\rangle$ to node $i$

        \vspace{0.5em} 
        \item \textbf{Upon} receiving $\langle\text{Mb-Ack},R(b_p^i),\sigma_j(b_p^i)\rangle$ \textbf{do}
        \item \hspace{1em}$\boldsymbol{Sigs}=\boldsymbol{Sigs}\cup \sigma_j(b_p^i)$
        \item \hspace{1em}\textbf{if} $|\boldsymbol{Sigs}|\geq 2f+1$ \textbf{do}
        \item \hspace{2em}$\sigma(b_p^j)\gets TComb(R(b_p^i),\boldsymbol{Sigs})$. 
        \item \hspace{2em}$C_p^i\gets (R(b_p^i),\sigma(b_p^i))$. \textcolor{blue}{\Comment{$C_p^i$ is the AC of $b_p^i$}}
        \item \hspace{2em}Store $C_p^i$
    \end{algorithmic}
\end{algorithm}

The AC organizes microblocks from the same node into a chained structure, facilitating the resolution of the sequencing issues of microblocks from honest nodes. 
Secondly, once an honest node collects an AC, it can be assured that all honest nodes will eventually obtain the complete microblock. 
Thus, during consensus, nodes can proceed without waiting for full content, avoiding delays due to missing microblocks.





\begin{algorithm}[t]
    \caption{Microblock Retrieval with $R(b_p^j)$ at $i$}
    \label{Alg:Microblock Ret}
    \begin{algorithmic}[1]
        \item Local variables:
        \item $\boldsymbol{Chk}\gets \{\}$  \textcolor{blue}{\Comment{Chunks corresponding to $R(b_p^i)$}}

        \vspace{0.5em}
        \item \textbf{Upon} event$\langle\text{Mb-Retrieval},R(b_p^j)\rangle$
        \item \hspace{1em}\textbf{if} there exists a local chunk $s_p^{j,i}$ with identifier $R(b_p^j)$
        \item \hspace{2em}Broadcast message $\langle\text{Mb-Chk}, R(b_p^j), s_p^{j,i}, P_i\rangle$
        \item \hspace{1em}\textbf{if} $R(b_{p-1}^j)$ has not yet been triggered for retrieval
        \item \hspace{2em}Trigger $\langle\text{Mb-Retrieval},R(b_{p-1}^j)\rangle$ 
            
        \vspace{0.5em} 
        \item \textbf{Upon} receiving $\langle\text{Mb-Chk},R(b_p^j),S_k,P_k\rangle$ for the first time
        \item \hspace{1em} \textbf{if} the microblock corresponding to $R(b_p^j)$ has yet not been decoded
        \item \hspace{2em}$\boldsymbol{Chk}=\boldsymbol{Chk}\cup S_k$
        \item \hspace{1em}\textbf{if} $|\boldsymbol{Chk}| = f+1$
        \item \hspace{2em}Trigger $\langle\text{Mb-Dec},R(b_p^j),\boldsymbol{Chk}\rangle$
        
        \vspace{0.5em}
        \item \textbf{Upon} event $\langle\text{Mb-Dec},R(b_p^j),\boldsymbol{Chk}\rangle$

        \item \hspace{1em}$b'\gets Dec(\boldsymbol{Chk})$ \textcolor{blue}{\Comment{Decode chunks $Chk$}}
        \item \hspace{1em}Mark $R(b_p^j)$ as available
        \item \hspace{1em}Compute $R(b')$
        \item \hspace{1em}\textbf{if} $R(b')\neq R(b_p^j)$ 
        \item \hspace{2em}Mark $R(b_p^j)$ as an empty microblock
    
    \end{algorithmic}
\end{algorithm}

\subsubsection{Retrieval Phase.}
\label{sec:retrieval detail}
When a node triggers the \(\langle\text{Mb-Retrieval}, R(b_p^j)\rangle\) event, it enters the retrieval phase for \(R(b_p^j)\). 
Upon entering the retrieval phase, node $i$ checks whether it has previously received a chunk \(s_i\) corresponding to \(R(b_p^j)\). If so, it broadcasts the message \(\langle\text{Mb-Chk}, R(b_p^j), s_i, P_i\rangle\). It then checks whether an Mb-Retrieval event for the predecessor microblock of \(b_p^j\) has ever been triggered; if not, it recursively triggers the retrieval event for the predecessor microblock. 

Each node will trigger the retrieval for each microblock at most once. This means that once consensus has been reached and the microblock has been executed, it can be safely removed from memory without concerning that other nodes will need to fetch contents of the microblock. 

Upon receiving an \(\langle \text{Mb-Chk}\rangle\) message corresponding to \(R(b_p^i)\), the node verifies the accompanying Merkle proof. If the verification is successful, the chunk is added to the local mempool; otherwise, the message is ignored.

When a node collects \(f+1\) chunks for \(R(b_p^j)\), decoding is performed to reconstruct \(b'\). The decoded\(b'\) is then re-encoded using \(Enc\) into a set of chunks \(\boldsymbol{Chk'}\) and its Merkle root is recomputed. If the recomputed Merkle root  matches \(R(b_p^j)\), decoding is successful. Otherwise, \(b_p^j\) is marked as empty (\texttt{Nil}). Future consensus on this microblock will regard it as empty.

Additionally, if a node receives \(2f+1\) Mb-Chk messages for \(R(b_p^i)\) before broadcasting its own chunk, it can consider the Retrieval of corresponding microblock as completed. Since at least \(f+1\) messages are from honest nodes and seen by others, all honest nodes will eventually reach the decoding threshold.

\subsubsection{Phase Triggering.}
We defer the discussion on when to trigger the Dispersal phase and Retrieval phase here. By default, the \(\langle\text{Mb-Dispersal}\rangle\) and \(\langle\text{Mb-Retrieval}\rangle\) events are triggered automatically, but when and how to trigger them depend on the specific requirements of BFT consensus.

A straightforward approach is to periodically trigger the $\langle\text{Mb-Dispersal}\rangle$ event and have nodes broadcast their AC(s) once they are collected. Upon receiving an AC message, the corresponding $\langle\text{Mb-Retrieval}\rangle$ event is triggered. Once all nodes have triggered the Retrieval event, they can decode the complete content of the microblock, thereby achieving totality. Consequently, each node can maintain a local mempool in a chained structure, allowing microblocks to be efficiently shared among nodes.

We will discuss in Section~\ref{Sec:bft-protocol} how we trigger the $\langle$Mb-Dispersal$\rangle$ and $\langle$Mb-Retrieval$\rangle$ events to meet our requirements when combined with BFT consensus.

\subsection{Analysis  of Imitater}
\label{sec:smp-communication}
\subsubsection{Communication Efficiency.}
To achieve a Shared Mempool (SMP) with the totality property, we employ coding techniques for microblock distribution. Here, we briefly analyze the efficiency of completing the distribution of a microblock.


Assume the microblock size is \( m \), and both signatures and hash values are of the same size $\lambda$ for simplicity. During the Dispersal phase, the distributor needs to encode the microblock into \( n \) chunks, each of size \( \frac{m}{f+1} \), with each Mb-Dis message containing an identifier (\( \lambda \)), an availability certificate (AC) (\( \lambda \)), a chunk (\( \frac{m}{f+1} \)), and a Merkle proof (\( \lambda \log n \)). When receiving an Mb-Dis message, other nodes send back an Mb-ack message containing a partial signature (\( \lambda \)). The total communication cost for this phase is \( 2n\lambda + n\lambda \log n + \frac{mn}{f+1} \). For large messages (\( m \gg n \)), the Dispersal cost approximates \( 3m \) with \( n = 3f + 1 \).



In the Retrieval phase, each node broadcasts a message with an identifier (\( \lambda \)), a chunk (\( \frac{m}{f+1} \)), and a Merkle proof (\( \lambda \log n \)). The total communication cost for this phase is \( \frac{mn^2}{f+1} + n^2\lambda \log n \), which approximates \( 3mn \) for large microblocks.


In summary, the total communication overhead for completing microblock distribution is \( 2n\lambda + (n^2+n)\lambda \log n + \frac{m(n^2+n)}{f+1} \), and for large microblocks, the cost is dominated by the Retrieval phase at \( 3mn \).

\subsubsection{Correctness.} Here we briefly show the correctness of Imitater.
\begin{lemma}[Dispersal Termination]
    \label{lemma:Dispersal_termination}
    Every honest node initiating microblock Dispersal will eventually obtain a corresponding AC after GST.
\end{lemma}

\begin{proof}
    Honest nodes distribute the Mb-Dis message to all nodes, and other honest nodes respond with Mb-Ack messages upon receipt. The sender can collect Mb-Ack messages from at least \(2f+1\) honest nodes and synthesize the final signature from the partial signatures to form the AC.
\end{proof}

While Lemma~\ref{lemma:Dispersal_termination} ensures that nodes engaging in honest distribution will collect a certificate, it is also essential to guarantee that all other honest nodes can obtain the corresponding microblock. Now, we show that any node could collect at most one valid AC while distributing microblock at position $p$.

\begin{lemma}[Microblock Uniqueness]
    \label{lemma:microblock uniqueness}
    For any position \( p \) on any mempool chain \( i \), at most one microblock can have a valid AC.

\end{lemma}

\begin{proof}
    Suppose node \( i \) distributes two microblocks, \( b_p^i \) and \( \hat{b}_p^i \), at position \( p \), and collects valid ACs \( C_p^i \) and \( \hat{C}_p^i \) for both. 
    Since an AC is formed by aggregating partial signatures from the ack messages of at least \( 2f+1 \) honest nodes, at least one honest node must have sent ack messages for both microblocks at the same position \( p \). This contradicts the protocol, as an honest node should only ack one microblock per position. Therefore, at most one valid AC can be formed.
    
\end{proof}

Next, we need to ensure that once an AC is formed from honest node, all other honest nodes can obtain the corresponding microblock.


\begin{theorem}[SMP Totality]
    \label{lemma:SMP_avaliablity}
    \label{lemma:Microblock consistency}
    If an honest node collects an AC for a microblock, then all honest nodes will eventually obtain a consistent copy of that microblock.
    
\end{theorem}

\begin{proof}
    We first show that any honest node will ultimately be able to decode a microblock that has the AC. As an AC comprises \(2f+1\) of Mb-Ack messages, of which at least \(f+1\) must originate from honest nodes, this indicates that at least \(f+1\) honest nodes possess distinct chunks. These honest nodes will broadcast their chunks when triggering the retrieval, allowing all honest nodes to decode.
    
    Since each chunk comes with a Merkle tree proof, a node will re-encode after decoding and generate a corresponding Merkle tree, then check whether the newly generated Merkle tree matches the original one. Hence, we require that either (i) every honest node's verification succeeds, or (ii) every honest node's verification fails. We can prove it by contradiction. Let's assume that for a set of segments $Chk$, there exist two subsets $S_1$ and $S_2$ of size $f+1$, corresponding to valid Merkle trees. Without loss of generality, we assume that decoding is successful for $S_1$ and fails for $S_2$.
    
    Since the decoding for $S_1$ is successful, let the decoding result be $b$. Thus, re-encoding $b$ will correctly yield $Chk' = Chk$. According to the principles of erasure codes, using any $f+1$ segments from $Chk'$ will result in consistent decoding. In other words, any subset of $Chk$ will obtain consistent encoding, which contradicts the fact that decoding failed for $S_2$.
    
    Hence, different honest nodes decoding chunks corresponding to the same Merkle root \(R(b_p)\) will obtain consistent results.
    
\end{proof}




   


Since our SMP is maintained in a chained structure, we also need to provide a chain-based consistency constraint for microblocks.




\begin{theorem}[SMP Chain-Consistency]
    \label{Theorem:SMP chain-consistency}
    For any microblock chain \(i\), where \(i \in [n]\) on the SMP, if any honest node $j$ and \(k\) have \(b_p^{i,j}\) and \(b_p^{i,k}\), respectively, and \(b_p^{i,j} = b_p^{i,k}\), then \(b_{p'}^{i,j} = b_{p'}^{i,k}\) for all \(p' \leq p\).


\end{theorem}


\begin{proof}
    Given that \(b_p^{i,j} = b_p^{i,k}\), it follows that the predecessor identifiers included within the microblock are also equal, that is, \(R(b_{p-1}^{i,j}) = R(b_{p-1}^{i,k})\). By Lemma~\ref{lemma:Microblock consistency}, we have \(b_{p-1}^{i,j} = b_{p-1}^{i,k}\), and this logic can be recursively applied.

\end{proof}

\section{Towards Practical BFT System}
\label{Sec:bft-protocol}

In this section, we describe how to utilize Imitater to implement a partially synchronous consensus protocol, namely Imitater-BFT protocol.

\subsection{Building BFT based on Imitater}
\label{sec:building-bft}
Imitater-BFT can be seen as a traditional BFT protocol where the leader's task of packaging transactions into blocks is replaced by packaging microblock identifiers. Each node then runs a Shared MemPool to share microblocks, resulting in a more efficient BFT protocol. Unlike other Shared MemPool-based protocols, Imitater-BFT follows a \emph{dispersal-consensus-retrieval} framework (Figure~\ref{fig:Imitater_BFT}), where the \emph{dispersal} and \emph{retrieval} phases are derived from Imitater’s Dispersal and Retrieval phases, respectively.

\begin{figure}
    \centering
    \includegraphics[width=0.7\linewidth]{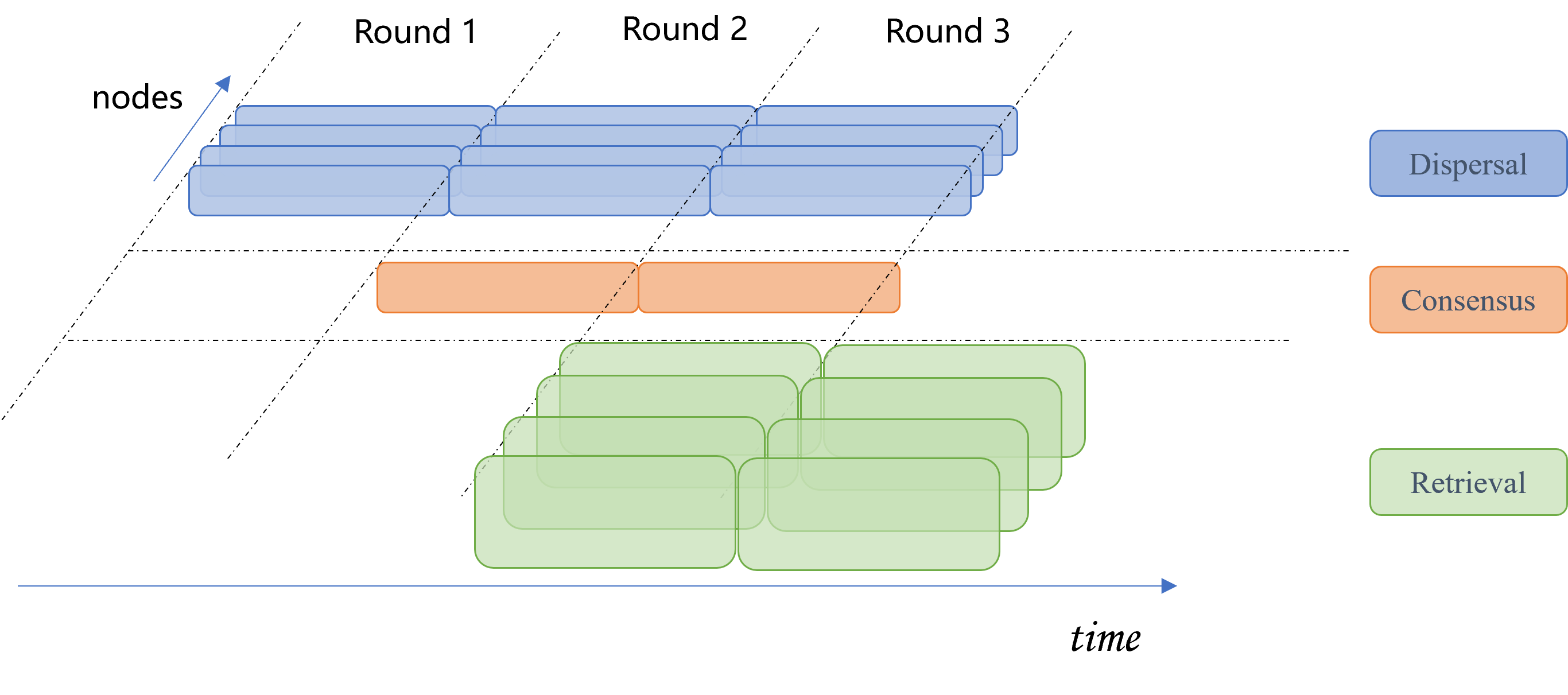}
    \caption{The process of integrating Imitater into consensus}
    \label{fig:Imitater_BFT}
\end{figure}


The Dispersal, Retrieval, and Consensus phases run in parallel. Nodes continuously perform microblock dispersal, collect the corresponding  ACs, and report the latest ACs to the leader. At the start of each round, the leader aggregates the reported ACs and microblock identifiers into a candidate block for consensus.

Once a candidate block reaches consensus, nodes initiate the Mb-Retrieval process for all microblocks referenced within it and wait for these microblocks to become available. As discussed in Section~\ref{sec:retrieval detail}, retrieving a microblock recursively triggers the retrieval of its predecessor microblocks that have not yet been retrieved. This ensures that the entire chain of blocks is fully available, maintaining protocol consistency.
Once all the microblocks within the block have been decoded and are available, the confirmation process for the block begins. Specifically, all requests within the microblocks are extracted and ordered based on position and timestamp to form a transaction list. Nodes then execute these transactions sequentially and return the results to the clients.

To mitigate the risk of a dishonest leader censoring certain nodes' microblocks, we adopt a consensus protocol with a built-in leader-rotation mechanism. 
We provide detailed steps in Appendix~\ref{sec:imitater-fhs} for compiling Imitater into Fast-HotStuff~\cite{jalalzai2023fast}, along with a formal proof of correctness.


\subsection{Bandwidth Fluctuation Adaptivity}
\label{sec:trigger-dispersal}  

Most BFT protocols are designed for stable network environments where communication bandwidth between nodes is assumed to be constant.  
However, in real-world scenarios, bandwidth often fluctuates over time. In conventional protocols that involve message distribution and vote collection, bandwidth fluctuations can lead to speed differences among nodes. As a result, faster nodes are forced to wait for slower ones, causing inefficiencies.  
For example, in Narwhal~\cite{danezis2022narwhal}, microblock distribution progresses in rounds, each requiring \(2f+1\) acknowledgment messages. The overall progress is effectively limited by the \(f+1\)-th slowest node, leading to underutilization of available bandwidth.

We consider an environment where the bandwidth of all nodes is governed by the same random process but varies independently across nodes.

\begin{definition}[Bandwidth Adaptivity]
\label{def:bandwidth-adapitivity}
There exists a constant \(\beta > 0\) such that, for any time \(t\), if the bandwidth \(B_i(t)\) of each node $i$ fluctuates independently around a common mean value and satisfies \(B_i(t) \ge \beta,  \forall t\), then the throughput of the Imitator-BFT protocol remains constant over time.
\end{definition}

\subsubsection{Benefits of Decoupling Dispersal and Retrieval.}
In Imitater-BFT, the Retrieval phase is triggered after a block is committed. However, the exact timing is flexible and can be adjusted based on deployment needs. Decoupling microblock distribution into separate Dispersal and Retrieval phases enables the system to better adapt to varying bandwidth conditions.

As analyzed in Section~\ref{sec:smp-communication}, Dispersal involves relatively low message volume (about \(3m\) for microblock size \(m\)) compared to Retrieval (\(3mn\)). During the consensus phase, the messages are relatively small, consisting only of identifiers and essential information.
By allocating stable bandwidth to Dispersal and Consensus messages and adjusting remaining resources for Retrieval, we eliminate the need to wait for the \(f+1\)-th slowest node during Dispersal as all nodes have sufficient bandwidth resources during the Dispersal phase, improving adaptability.  

Assuming that for each node \(i\), the available bandwidth \(B_i(t)\) has a known and consistent mean across time, the dispersal rate can be set proportional to \(1/n\) of the average bandwidth, as the size of dispersal messages in one round is approximately \(1/n\) of the retrieval messages size. Hence, Definition~\ref{def:bandwidth-adapitivity} was satisfied.

\subsubsection{Dynamic Dispersal Triggering.} 
In more general scenarios, average bandwidth is typically unknown or fluctuates, making hardcoded rates impractical. To alleviate fluctuating bandwidth, we propose a mechanism for dynamically adapting the Dispersal phase in Imitater-BFT. 
We dynamically adjust the dispersal rate by monitoring the retrieval queue. Each dispersal of a microblock corresponds to a retrieval event. By observing the gap between the number of dispersals (\(N_d\)) and completed retrievals (\(N_r\)), we infer network conditions:
\begin{itemize}
    \item If the gap is small, indicating sufficient bandwidth, the dispersal rate is increased by reducing the interval \(\tau\) between events.  
    \item If the gap is large, suggesting overload, \(\tau\) is increased to reduce the dispersal rate. 
\end{itemize}

\begin{algorithm}[t]  
\caption{Triggering Mb-Dispersal for Node \(i\)}  
\label{Alg:triggerMBD}  
\begin{algorithmic}[1]  
\item Local Variables:  
\item \(\tau \gets 0\) \textcolor{blue}{\Comment{Time interval between dispersal events}}  
\item \(N_d, N_r \gets 0\) \textcolor{blue}{\Comment{Completed Dispersals and Retrievals}} 

\vspace{0.5em}
\item \textbf{while} (true) \textbf{do}  
\item \hspace{1em}Sleep($\tau$) \textcolor{blue}{\Comment{Control Dispersal rate}}
\item \hspace{1em}$b_h^i\gets \langle i,h,\boldsymbol{t} \rangle$ \textcolor{blue}{\Comment{Extract a set of requests $\boldsymbol{t}$} from pending requests}
\item \hspace{1em}Wait until $C_{h-1}^i$ formed
\item \hspace{1em}Trigger $\langle \text{Mb-Dispersal},b_h^i,C^i_{h-1}\rangle$
\item \hspace{1em}\textbf{if} \(N_d - N_r \geq t\): \(\tau \gets \tau + \alpha\) \textcolor{blue}{\Comment{Increase interval}}  
\item \hspace{1em}\textbf{else}: \(\tau \gets \tau - \alpha\) \textcolor{blue}{\Comment{Decrease interval}}  
\end{algorithmic}  
\end{algorithm}  

Algorithm~\ref{Alg:triggerMBD} outlines the dynamic triggering mechanism, which adapts dispersal frequency based on the disparity between \(N_d\) and \(N_r\). The metrics \(N_d\) (dispersals) and \(N_r\) (completed retrievals) guide the adjustment of \(\tau\), ensuring efficient bandwidth utilization without overloading the system. Since retrieval operates asynchronously, it can be offloaded to auxiliary machines for enhanced scalability, as demonstrated in Narwhal~\cite{danezis2022narwhal}.

\subsection{Practical Problems}
\label{sec:practical-problems}
Here we discuss some practical challenges encountered in BFT systems.
\subsubsection{Order Keeping.}

Maintaining the order of transactions is often overlooked in prior work, yet it is crucial when transactions depend on one another. For example, in the UTXO model, \( tx' \) may depend on funds from \( tx \), or in the account model, \( tx \) might enable \( tx' \) by making the account balance positive. If two dependent transactions are delegated to the same node but included in different microblocks, a Byzantine leader could withhold the microblock containing \( tx \) while allowing \( tx' \) to reach consensus first, rendering \( tx' \) invalid.
Thus, the final execution order must reflect the transaction submission order.

\begin{definition}
    \textbf{Order Keeping.} For any two transactions, $tx$ and $tx'$, submitted by a client to an honest node, the order in which they are submitted must be preserved in the final execution order.
\end{definition}


In Imitater-BFT, if \( tx \) and \( tx' \) are submitted sequentially by a client to an honest node \( r \), and placed in the same microblock, their execution order will naturally match the submission order. If placed in different microblocks \( b \) and \( b' \), \( b \) will precede \( b' \) since honest nodes package transactions in order. In case \( b \) is missed, \( b' \) will implicitly include \( b \) when added to a block, ensuring that \( b \) is committed before or simultaneously with \( b' \). Thus, the transaction order submitted to honest nodes is preserved in the final execution.

\subsubsection{Over Distribution.}
In decoupled architectures, some protocols\cite{gai2023scaling,hu2022leopard,hu2023data} allow nodes to independently distribute blocks without round-based synchronization constraints. While this provides flexibility, it also exposes the system to the risk of flooding attacks. Malicious nodes can flood the network with an excessive number of data blocks, forcing honest nodes to store non-consensus-related blocks in memory, which increases the risk of memory overflow.

In Imitater-BFT, malicious nodes might engage solely in dispersal, disregarding retrieval. By using up available bandwidth to continually perform dispersal, they increase the retrieval burden on other nodes, while decreasing their own dispersal volume. This doesn't reduce total throughput but skews the proportion of malicious microblocks in the consensus block, compromising fairness.

We enforce a rule that the position of any microblock distributed by node \( i \) must not exceed a preset threshold \( k \) relative to the position of its latest committed microblock.  If exceeded, nodes withhold ack responses until the position is within the threshold, preventing malicious flooding and network congestion.

\subsubsection{Unbalanced Workload.}
In practical deployment scenarios, clients may send transaction requests to nearby or trusted nodes, causing uneven workloads across nodes. This can lead to some nodes being overloaded while others remain idle. Stratus~\cite{gai2023scaling} mitigates this by allowing overloaded nodes to outsource microblock distribution tasks to idle nodes. However, this introduces additional latency and may require multiple outsourcing requests to complete the task.

In contrast, the Imitater-BFT protocol handles this issue more efficiently. During microblock distribution, the retrieval phase’s communication load is evenly distributed across all nodes, ensuring the distributor shares a similar load as the others. Idle nodes can participate by packaging and distributing empty microblocks, helping to balance the workload and alleviates overloaded nodes.

\section{Evaluation}
\label{Sec:evaluation}
\subsection{Implementation and Experimental Setup}
We implemented a prototype of the Imitater protocol in Golang, utilizing threshold signatures\footnote{https://github.com/dfinity-side-projects/bn}~\cite{boneh2001short} and employing Reed-Solomon codes\footnote{https://github.com/templexxx/reedsolomon}~\cite{reed1960polynomial} for erasure coding.



To evaluate its performance, we compared it with Stratus, a state-of-the-art protocol, using a modified version of bamboo-stratus\footnote{https://github.com/gitferry/bamboo-stratus}, where we replaced the memory pool logic with Imitater. The comparison was made with Stratus-HS from bamboo-stratus, which, to our knowledge, demonstrates superior throughput as the number of replicas increases. The consensus part was HotStuff.

Experiments were conducted on Cloud SA5.2XLARGE16 instances within a single datacenter, each with 8 vCPUs and 3 Gbps of internal bandwidth. Each replica ran on a separate EC2 instance, and we simulated a WAN environment with 100 Mbit/s replica bandwidth using \emph{tc}.


Four instances ran client processes, continuously sending requests to the replicas. Latency was measured as the time from when a node (the microblock proposer) receives a request to when it completes the consensus process. Throughput is the number of requests committed per second, averaged across all replicas. All measurements were taken after the system's performance had stabilized.
\subsection{Performance}

To evaluate the performance of Imitater, we conducted a thorough assessment of Imitater and compared with Stratus, the SOTA SMP protocol.

\begin{figure}[!th]
    \centering
    \begin{tikzpicture}
        \begin{axis}[
        width=7cm,
        height=4.5cm,
        xlabel={Number of faulty nodes $f$},
        ylabel={Throughput (Kops/s)},
        xmin=0, xmax=16,
        ymin=0, ymax=70,
        xtick={0,1,2,3,4,5,8,12,16},
        ytick={10,30,50},
        grid=both,
        grid style={dashed},
        major tick style={thick,black},
        xtick pos=bottom,
        ytick pos=left,
        tick align=center,
        legend style={
                fill=white,
                fill opacity=0.5,
                draw opacity=0.5,
                text opacity=1
            },
        ]   
        \addplot[blue,thick,mark=triangle,opacity=0.7] coordinates {
            (0,59.3)(1,30.5)(2,22.8)(3,17.4)(4,14.7)(5,12.5)(6,10.3)(7,8.8)(8,8)
            (9,7.7)(10,6.6)(11,6.2)(12,5.7)(13,5.4)(14,5)(15,4.8)(16,4.3)
        };
         \addplot[red,thick,mark=square,opacity=0.7] coordinates {
            (0,25.4)(1,25.4)(2,25.4)(4,25.4)(8,25.4)(16,25.4)
        };
        \addlegendentry{Stratus-HS}
        \addlegendentry{Imitater-HS}
    \end{axis}
    \end{tikzpicture}
    \caption{Throuput with $n$=49 and bandwidth settiing to 100Mbps, varying faulty nodes number $f$}
    \label{fig:Stratus-HS_with_faulty_nodes}
\end{figure}
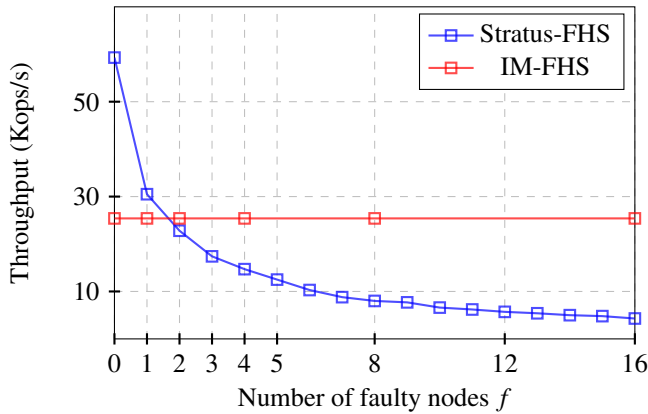


First, we analyzed the impact of faulty nodes on throughput. Figure~\ref{fig:Stratus-HS_with_faulty_nodes} shows that Stratus-HS throughput declines sharply as the number of faulty nodes increases from 0 to 16, in a 49-node system with 100Mbps bandwidth. For instance, just one faulty node reduces throughput by nearly half. In contrast, Imitater-HS maintains constant throughput regardless of faulty nodes. With one faulty node, Imitater-HS achieves 0.8$\times$ the throughput of Stratus-HS, and with 16 faulty nodes, it reaches 6$\times$ the throughput.



\begin{figure}[!ht]
    \centering
    \begin{subfigure}{0.48\textwidth}
        \centering
        \begin{tikzpicture}
            \begin{axis}[
            width=\textwidth,
            height=5cm,
            xlabel={Number of nodes $n$},
            ylabel={Throughput (Kops/s)},
            xmin=0, xmax=260,
            ymin=0, ymax=50,
            xtick={4,25,49,100,256},
            ytick={10,30,50},
            grid=both,
            grid style={dashed},
            major tick style={thick,black},
            xtick pos=bottom,
            ytick pos=left,
            tick align=center,
            legend style={
                fill=white,
                fill opacity=0.5,
                draw opacity=0.5,
                text opacity=1
            },
            ]
            \addplot[blue,thick,mark=triangle,opacity=0.7] coordinates {
                (4,48.3)(7,28.7)(16,12.8)(25,8.5)(49,4.3)(100,2.3)
            };
            \addplot[red,thick,mark=square,opacity=0.7] coordinates {
                (4,40.1)(7,35.4)(16,27.2)(25,26.5)(49,25.4)(100,21.5)(200,17)(256,14.8)
            };
            \addplot[black,dashed,mark=triangle,opacity=0.9] coordinates {
                (100,2.3)(200,0)
            };
            \addlegendentry{Stratus-HS}
            \addlegendentry{Imitater-HS}
            \end{axis}
        \end{tikzpicture}
        \caption{Throughput vs. number of nodes.}
        \label{fig:tp_vs_nodes}
    \end{subfigure}
    \hfill
    \begin{subfigure}{0.48\textwidth}
        \centering
        \begin{tikzpicture}
            \begin{axis}[
            width=\textwidth,
            height=5cm,
            xlabel={Number of nodes $n$},
            ylabel={Latency (ms)},
            xmin=4, xmax=260,
            ymin=0, ymax=10000,
            ymode=log,
            xtick={4,25,49,100,256},
            ytick={100,1000,10000},
            grid=both,
            grid style={dashed},
            major tick style={thick,black},
            xtick pos=bottom,
            ytick pos=left,
            tick align=center,
            legend style={
                fill=white,
                fill opacity=0.5,
                draw opacity=0.5,
                text opacity=1
            },
            ]
           \addplot[blue,thick,mark=triangle,opacity=0.7,error bars/.cd,
                    y dir=both,
                    y explicit,
                    error bar style={thick}] coordinates {
                (4,33)+-(0,2)
                (7,77)+-(0,10)
                (13,184)+-(0,15)
                (16,262)+-(0,50)
                (25,628)+-(0,150)
                (49,2250)+-(0,250)
                (100,5565)+-(0,500)
            };
            
            \addplot[red,thick,mark=square,opacity=0.7,
                    error bars/.cd,
                    y dir=both,
                    y explicit,
                    error bar style={thick}] coordinates {
                (4,37)+-(0,5)
                (7,66)+-(0,10)
                (16,163)+-(0,50)
                (25,279)+-(0,150)
                (49,324)+-(0,200)
                (100,423)+-(0,250)
                (200,617)+-(0,260)
                (256,800)+-(0,300)
            };
           
            \addplot[black,dashed,mark=triangle,opacity=0.9] coordinates {
            (100,5565)(200,10000)
            };
            \addlegendentry{Stratus-HS}
            \addlegendentry{Imitater-HS}
            \end{axis}
        \end{tikzpicture}
        \caption{Latency vs. number of nodes.}
        \label{fig:lat_vs_nodes}
    \end{subfigure}

    \caption{Performance comparison of Imitater-HS and Stratus-HS with bandwidth setting to 100 Mbps and up to $n/3$ nodes being faulty. (a) Throughput vs. nodes. (b) Latency vs. nodes.}
    \label{fig:performance_comparison}
\end{figure}

Figure~\ref{fig:tp_vs_nodes} demonstrates how throughput varies with the number of nodes under maximum fault tolerance, i.e., the number of faulty nodes is $n/3$. The throughput of Stratus-HS decreases much faster than Imitater-HS as the number of nodes increases. 
At 100 nodes, Imitater-HS achieves 21.5 Kops/s, nearly 9$\times$ that of Stratus-HS (2.3 Kops/s). 
The throughput decline of Stratus-HS is primarily due to all malicious nodes requesting all microblocks from each honest node, which consumes the honest nodes' bandwidth by forcing them to send these redundant microblocks, thereby impacting throughput. Thanks to the design of Imitater-HS, malicious nodes cannot make malicious requests, thus avoiding any performance degradation. In contrast, as the number of nodes increases under maximum fault tolerance, the number of malicious nodes also rises, leading to a noticeable performance decline in Stratus-HS.


Similarly, the latency in Stratus-HS is affected by attacks from malicious nodes. Figure~\ref{fig:lat_vs_nodes} shows how latency changes as the number of nodes increases. It can be observed that Imitater-HS maintains a relatively low latency, while Stratus-HS experiences a sharp increase in latency as the number of nodes grows. Specifically, with 100 nodes, Stratus-HS's latency reaches 5565 ms, whereas Imitater-HS's latency remains around 550 ms.

\section{Conclusion}
\label{Sec:conclusion}
In this paper, we propose a novel SMP protocol, namely Imitater, to improve the efficiency of microblock distribution with Byzantine faulty nodes under partially synchronous setting. Imitater ensures that all microblocks are correctly distributed, available, and ordered even in the presence of faulty nodes. Imitater is easy to integrate into a BFT protocol achieving improved performance by address specific practical problems and better bandwidth use. Our experiments, conducted in a large-scale node deployment environment, showed that the protocol is efficient, maintaining high throughput and low latency, even when some nodes are faulty.  


\begin{credits}
\subsubsection{\ackname}
This work was supported by National Natural Science Foundation of China (Grant 62301190), Shenzhen Colleges and Universities Stable Support Program (Grant GXWD20231129135251001), National Key Research and Development Program of China (Grant 2023YFB3106504) and Major Key Project of PCL (Grant PCL2023A09).

\subsubsection{\discintname}
The authors have no competing interests to declare that are
relevant to the content of this article.

\end{credits}
%
%

\bibliographystyle{splncs04}
\bibliography{ref}

\appendix
\section{Integrating Imitater into Fast-HotStuff}
\label{sec:imitater-fhs}

In this section, we illustrate how to utilize Imitater to implement a partially synchronous consensus protocol, namely Imitater-BFT protocol.

\subsection{Overview}


As shown in Figure~\ref{fig:Imitater_BFT}, Imitater family BFT consensus works in a \emph{disperse-consensus-retrieval} framework, where the \emph{disperse} and \emph{retrieval} parts are derived from the Dispersal and Retrieval phases, respectively.  
In each round of consensus, all nodes disperse microblocks, collect corresponding ACs and report them to the leader. The leader packages the latest ACs reported by all nodes, along with the corresponding microblock identifiers, into a block for consensus. When the block carrying ACs is broadcast, the retrieval part can be triggered.


To prevent a malicious leader from censoring certain nodes' microblocks by deliberately excluding some fully dispersed microblocks from the proposal, we integrate Imitater into a leader-rotation BFT framework. 
Since our SMP is maintained in a chain structure where each microblock contains the AC of its predecessor, all uncommitted predecessor blocks on the chain of a microblock included in a block are implicitly included in that block as well. With the leader-rotation mechanism, as long as an honest node is elected as the leader, even if certain microblocks were censored by a malicious leader in a previous round, subsequent honest leaders can include these censored or omitted microblocks in the consensus process without additional cost. This design ensures 2/3 \emph{chain quality}, meaning that at least 2/3 of the final consensus content originates from honest nodes.


\subsection{Detail of Imitater-FHS}
We now take Fast-HotStuff~\cite{jalalzai2023fast}, an improved version of HotStuff~\cite{yin2019hotstuff}, as an example and present the details of how we integrate Imitater into a BFT protocol. Imitater-FHS operates in a pipelined manner, with the leader changing in each round. Note that Imitater is compatible with most SMP-based BFT protocols.

\begin{algorithm}[!th]
\caption{Utilities} 
\label{alg:utilities}
\begin{algorithmic}[1] 
\Procedure{CreateBlock}{$v$, $qc$, $aggQc$,$B'$}
\item \hspace{1em} $B.view \gets v$
\item \hspace{1em} $B.qc \gets qc$
\item \hspace{1em} $B.aggQc \gets aggQc$
\item \hspace{1em} $B.parent \gets B'$
\item \hspace{1em} \textbf{for} $i\in [n]$
\item \hspace{2em} $B.mbs\gets B.mbs\cup hC^i $
\item \hspace{1em} \textbf{Return} $B$
\EndProcedure
\vspace{0.5em}
\Procedure{CreateQC}{$v$, $\boldsymbol{V_s}$}
\item \hspace{1em}$qc.view\gets v$
\item \hspace{1em}$qc.block\gets m.block:m\in \boldsymbol{V_s}$
\item \hspace{1em}$qc.sigs\gets Tcomb (qc.block,\{m.sig|m\in\boldsymbol{V_{s}} \})$ 
\item \hspace{1em}\textbf{Return} $qc$
\EndProcedure
\vspace{0.5em}
\Procedure{CreateAggQC}{$v$, $\boldsymbol{N_s}$}
\item \hspace{1em}$aggQc.qc_{set}$$\gets$Extract Quorum Certificates from $\boldsymbol{N_s}$
\item \hspace{1em}\textbf{Return} $aggQc$
\EndProcedure
\vspace{0.5em}
\Procedure{SafeProposal}{$P$,$v$}
\item \hspace{1em}$B\gets P.block$
\item \hspace{1em}\textbf{for} $C^j\in B.mbs$
\item \hspace{2em}\textbf{if} $C^j$ is not valid
\item \hspace{3em}\textbf{Return} False.
\item \hspace{1em}\textbf{if} $B.qc$
\item \hspace{2em}\textbf{Return} $(B.view \geq v)\wedge (B.view=qc.view+1)$
\item \hspace{1em}\textbf{if} $B.aggQc$
\item \hspace{2em}$hqc\gets$ highest QC in $aggQc$
\item \hspace{2em}\textbf{Return} $B$ extends from $hqc$

\EndProcedure

\end{algorithmic}
\end{algorithm}

\subsubsection{Utilities}
Before detailing Imitater-FHS, we introduce some useful utilities, as shown in Algorithm~\ref{alg:utilities}.

\textbf{View.} The protocol operates in rounds, each referred to as a \emph{view}. Each view has a designated leader and is assigned a view number. A deterministic algorithm is used to select the leader, following a round-robin scheme, so that every node knows which node is the leader for the current view. The process of switching the leader is called view change. Note that even in the normal case, where the leader performs honestly, view change occurs due to the leader-rotation scheme.

\textbf{Quorum Certificate (QC) and Aggregated QC.} 
When the leader starts a proposal, the leader collects a set $V_s$ consisting of \(n-f\) Vote messages together with partial signatures from the previous view, it runs procedure \emph{CreateQC}, using \(TComb\) to aggregate the partial signatures into a \emph{Quorum Certificate} (QC). A QC proves that at least $f+1$ honest nodes receive the block in the normal case.


If the leader from the previous view fails, causing the nodes to not vote properly and triggering the view change, all the nodes send New-View messages and their known highest QCs to the new leader. The new leader starts the \emph{CreateAgg} procedure to generate an \emph{Aggregated QC} (AggQC) by concatenating $2f+1$ QCs upon receiving a set $N_s$ consisting of $2f+1$ new view messages. An AggQC proves that at least $f+1$ honest nodes admit AggQC as the highest QC.



\textbf{Block Structure.} For proposing a block, the leader runs the \emph{CreateBlock} procedure to construct a \emph{block}. The block includes the view number \( v \) for the current view, a QC of the previous block (or an AggQC), the latest position ACs reported by all nodes and the hash of the parent block. By default, the first block is empty.

\textbf{SafeProposal.} Upon receiving the proposal message, node \(i\) executes the \emph{SafeProposal} procedure to check the validity of proposal $P$ in current view. This involves verifying the legitimacy of all included ACs and assessing whether the block is derived from 1) the highest QC in the normal case or 2) the highest QC in AggQC when the previous leader fails.

\subsubsection{Imitater-FHS Protocol}

By combining the utilities in Algorithm~\ref{alg:utilities}, we obtain Imitater-FHS, as shown in Algorithm~\ref{Alg:Imitater_BFT}. Note that we omit the Dispersal phase of SMP in Algorithm~\ref{Alg:Imitater_BFT} as 1) it can be executed in parallel with Algorithm~\ref{Alg:Imitater_BFT}, 2) triggering it depends on the bandwidth adaptability requirements and involves many details. We leave the discussion on triggering the Dispersal phase in Section~\ref{sec:trigger-dispersal}. 

\begin{algorithm}[!th]
\caption{Imitater-FHS Protocol (for node $i$)} 
\label{Alg:Imitater_BFT}
\begin{algorithmic}[1] 
\item As a leader
\item \textbf{Upon} entering view $v$
\item \hspace{1em}\textbf{If} a set $V_{set}$ consisting of $n-f$ Vote messages together with partial signatures is received
    \item \hspace{2em}$qc\gets$ CreateQC($v$,$V_{set}$)
    \item \hspace{2em}$B\gets$CreateBlock($v$,$qc$,$\perp$,$qc.block$)
\item \hspace{1em}\textbf{If} a set $N_{set}$ consisting of $n-f$ New-View messages is received
    \item \hspace{2em}$aggQc \gets$ CreateAggQC($v$, $N_{set}$)
    \item \hspace{2em}$hqc \gets$ highest QC in $aggQc$
    \item \hspace{2em}$B \gets$ CreatBlock($v$,$\perp$,$aggQc$,$hqc.block$)
\item \hspace{1em}Broadcast $\langle \text{Proposal},v,B\rangle_i$
\vspace{0.5em}
\item As a non-leader node
\item \textbf{Upon} receiving proposal message $P$ from leader
    \item \hspace{1em}\textbf{if} SafeProposal($P$,$curView$)
        \item \hspace{2em}$hC^i\gets$ highest AC in microblock chain $i$
        \item \hspace{2em}$sig\gets$ Sign on $P.B$
        \item \hspace{2em}Send $\langle \text{Vote},sig\rangle$ and $hC^i$ to leader $(v+1)$
    \item \hspace{2em}$B'\gets B.parent$
    \item \hspace{2em}$B''\gets B'.parent$
    \item \hspace{2em}\textbf{if} $B'.view = B''.view+1$
        \item \hspace{3em}\textbf{for} $C^j$ \textbf{in} $B''.mbs$
            \item \hspace{4em}$I(b)\gets C^j.id$
            \item \hspace{4em}Trigger $\langle\text{Mb-Retrieval},I(b)\rangle$
        \item \hspace{3em}\textbf{Until} all microblocks in $B''$ are available
            \item \hspace{4em}Sort all the requests
            \item \hspace{4em}Execute all requests
            \item \hspace{4em}Respond to clients
\vspace{0.5em}
\item \textbf{Timeout}
\item \hspace{1em}$hqc \gets$ highest QC in local
\item \hspace{1em}$hC^i\gets$ highest AC in microblock chain $i$
\item \hspace{1em}Send $\langle\text{New-View}, hqc, v+1 \rangle_i$ and $hC^i$ to the leader $(v+1)$
\item \hspace{1em}Enter view $v+1$
\end{algorithmic}
\end{algorithm}

In each round of consensus, the leader constructs a QC \(qc\) using collected Vote messages or New-View messages to prove that its proposal is secure.
Then the leader constructs a block based on the aforementioned \(qc\) (derived from Vote messages) or \(aggQc\) (derived from New-View messages) by call of CreatBlock procedure in Algorithm~\ref{alg:utilities}. The leader then broadcast the block.



As a non-leader node, upon receiving a proposal, node $i$ execute SafeProposal to verify its validity. If the verification is passed, the node signs on the block and sends a Vote message to the leader of view $v+1$, also attaching the latest AC for the \(i\)th chain. Then node $i$ enters view $v+1$.


If the parent block \(B' = B.parent\) and its grandparent block \(B'' = B'.parent\) have consecutive view numbers, then the grandparent block \(B''\) can be safely committed. Node \(i\) will initiate the Mb-Retrieval event of all microblocks within \(B''\) and wait for these microblocks to become available. As is discussed in Section~\ref{sec:retrieval detail}, the retrieval of a microblock will recursively trigger the retrieval of its predecessor microblocks which have never been triggered before. This ensures that the chain of blocks is fully available for processing and adherence to the protocol’s consistency requirements.


Once all the microblocks in block \(B\) have been decoded and are available, the confirmation process for the block can begin. Specifically, all requests contained within these microblocks are extracted and sorted by their microblock position and timestamp to form a transaction list. Following this, the node executes all transactions in the transaction list in sequence and returns response values to the clients.

Each node would set a timer when entering a new view. When the local timer expires, node $i$ sends its highest local QC along with the latest AC for the \( i \)th chain in a New-View message to the leader of view \( v+1 \) and enters view $v+1$.

\subsection{Analysis}
\subsubsection{Safety and Liveness}
\label{sec:bftanalysis}
Imitater-FHS ensures security, including safety and liveness, justified in Theorem~\ref{theorem:safety} and Theorem~\ref{theorem:liveness}, respectively.


\begin{theorem}{Safety:}
    \label{theorem:safety}
 No equivocated requests will be committed by honest nodes at the same position.
\end{theorem}
Here, we provide a brief proof and the complete proof is provided in Appendix~\ref{sec:full-proof}.

Requests are organized into microblocks and distributed among nodes. During consensus, the identifiers of microblocks are included in the proposal, and the confirmation of requests is achieved by confirming the proposal block. Thus, we need to prove: (i) All honest nodes commit a consistent proposal block in each round (Lemma~\ref{lemma:unique-commit}), and (ii) they can obtain a consistent transaction list from a consistent proposal (Lemma~\ref{lemma:Block consistency}).

For point (i), the proof aligns with Fast-HotStuff, which ensures that all honest nodes commit a consistent block. This consistency in microblock identifiers is inherited from the guarantees provided by Fast-HotStuff.
For point (ii), a microblock may be explicitly or implicitly included in a block. Since the block includes the AC for each microblock, all honest nodes can access consistent content for explicitly included microblocks (Lemma~\ref{lemma:Microblock consistency}). Furthermore, all microblocks that are implicitly included in the proposal can be retrieved as well (Theorem~\ref{Theorem:SMP chain-consistency}). A consistent sorting algorithm ensures that all honest nodes confirm requests in the same order, leading to consistent execution.

\begin{theorem}{Liveness:}
    \label{theorem:liveness}
    After GST, the confirmation for a pending request will always be reached.
\end{theorem}

\begin{proof}
Under a leader-rotation rule, for Byzantine nodes to prevent protocol liveness, they would need to interrupt the formation of a one-chain by ensuring a Byzantine leader follows each honest leader. In a system with \(n = 3f + 1\) nodes, where there are \(f\) Byzantine nodes, in the worst-case scenario, each could be strategically placed after an honest node. However, there would still be an additional \(f + 1\) honest nodes capable of forming a one-chain. Therefore, after GST, there would at least be two consecutive leaders forming a direct-chain, and a decision will eventually reach.

Given that a chain-based shared mempool (SMP) is employed and a Round-robin election method is used, microblocks from all honest nodes will eventually be included in some proposal. Consequently, every request will ultimately be executed. 
\end{proof}

\section{Proof of Theorem~\ref{theorem:safety}}
\label{sec:full-proof}
We define a sequence of blocks with consecutive views as a \emph{direct chain}. Specifically, if there are two blocks \(B\) and \(B'\) such that \(B'.view = B.view + 1\) and \(B'.parent = B\), then \(B\) and \(B'\) are said to form an \emph{one-direct chain}. Similarly, if there exists a block \(B''\) such that \(B''.view = B'.view + 1\), \(B'.view = B.view + 1\), \(B''.parent = B'\), and \(B'.parent = B\), then \(B''\), \(B'\), and \(B\) together constitute a \emph{two-direct chain}.

We define two blocks \(B\) and \(B'\) as \textbf{conflicting} if \(B\) is neither a predecessor nor a successor of \(B'\). This notion of conflict naturally extends to Quorum Certificates; specifically, \(qc1\) and \(qc2\) are considered conflicting if \(qc1.block\) and \(qc2.block\) are \textbf{conflicting}.

\begin{lemma}
    \label{lemma:non-comflicating qc}
    If any two valid Quorum Certificates, \(qc1\) and \(qc2\), are conflicting, then it must be the case that \(qc1.view \neq qc2.view\).
\end{lemma}

\begin{proof}
    We will prove this by contradiction. Assume that in the same view \(v\), two conflicting blocks \(B\) and \(B^*\) both receive sufficient votes and respectively form \(qc\) and \(qc^*\). Since a QC requires \(2f+1\) vote messages, it implies that at least one honest node voted twice in view \(v\), which is a contradiction because honest nodes vote at most once in each view. Therefore, for any two conflicting QCs, \(qc\) and \(qc^*\), it must be true that \(qc.view \neq qc^*.view\).
\end{proof}

\begin{lemma}
    \label{lemma:non-conflicating block}
    If \(B\) and \(B^*\) are conflicting, then only one of them can be committed by an honest replica    
\end{lemma}
\begin{proof}
From Lemma~\ref{lemma:non-comflicating qc}, we know that \(B.view \neq B^*.view\). Assuming block \(B\) has been committed by an honest node \(r\), at least \(f+1\) honest nodes are aware of its corresponding QC, \(qc\). When node \(r\) receives a proposal for block \(B^*\), which is conflicting with \(B\), the highQC in the proposal points to \(B\)'s ancestor block \(B^\circ\), such that \(B^*.qc.block = B^\circ\). If we consider \(B^\circ\) as the parent block of \(B\) (i.e., \(B.parent = B^\circ\)), then the AggQC of \(B^*\) will be deemed invalid. This is because any QC with a view number greater than or equal to \(qc\) should be included in the AggQC, yet the highQC included is for \(B\)'s parent, \(qc^\circ\), not \(qc\). Thus, this proposal will fail the SafeProposal check.
\end{proof}

Next, we demonstrate that if a node commits a block \(B\) at view \(v\), then all other honest nodes will also commit the same block at height \(v\).
\begin{lemma}
    \label{lemma:unique-extend}If an honest node $r$ commits a block \(B\), then $B$'s QC, \(qc\) will be used as the highQC for the next block \(B'\).
\end{lemma}
\begin{proof}
When the primary node starts a new view \(v'\), where \(v' > v\), with \(2f+1\) new-view messages and proposes block \(B'\), we need to demonstrate that among any set of \(2f+1\) new-view messages, at least one will contain the QC of block \(B\), \(qc\), or a successor QC of $qc$.

In view \(v\), node \(r\) committed block \(B\), meaning that at least \(2f+1\) nodes voted for \(B\)'s QC, implying that at least \(f+1\) honest nodes are aware of \(qc\)'s existence. Therefore, in any view greater than \(qc'.view\), when the primary node collects \(2f+1\) new-view messages, at least one of these messages, originating from one of these \(f+1\) honest nodes, will include \(qc\) or a QC inherited from \(qc\). Hence, block \(B'\) will have \(B\) as an ancestor.
\end{proof}

Next, we will prove that if an honest node commits block \(B\), then all other honest nodes will not commit any conflicting block \(B^*\).

\begin{figure}
    \centering
    \includegraphics[width=0.45\textwidth]{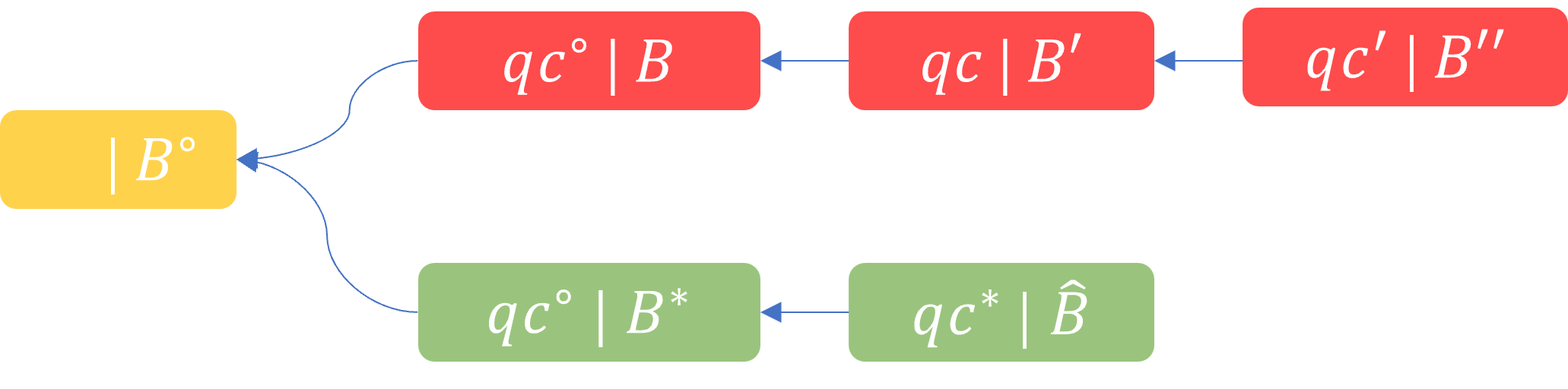}
    \caption{$B$ and $B^*$ both getting committed (impossible).}
    \label{fig:chain-fork}
\end{figure}

\begin{lemma}
    \label{lemma:unique-commit}
    It is impossible for any two conflicting blocks, \(B\) and \(B^*\), to be committed each by an honest node.
\end{lemma}

\begin{proof}
    We will prove this by contradiction. Suppose there is a fork in the blockchain, as shown in Figure~\ref{fig:chain-fork}. Without loss of generality, let us assume that an honest node \(i\) commits block \(B\) on the upper branch chain, and that another honest node \(i'\) commits a conflicting block \(B^*\).

    Let's analyze the different possible scenarios under this assumption involving \(B^*\)'s quorum certificate \(qc^*\). As \(qc^*\) is in conflict with both \(qc\) and \(qc'\) as indicated by Lemma~\ref{lemma:non-comflicating qc}, thus \(qc^*.view \neq qc.view\) and \(qc^*.view \neq qc'.view\).

    Case 1: \(qc^*.view < qc.view\). In this case, \(qc^*\) will not be chosen by any honest node as a highQC, because there exists a QC, \(qc\), with a view number higher than \(qc^*.view\).

    Case 2: \(qc.view < qc^*.view < qc'.view\). Since \(B\) was committed through a two-direct chain, it follows that \(qc.view + 1 = qc'.view\). Therefore, there cannot exist a \(qc^*\) whose view number lies between \(qc\) and \(qc'\).

    Case 3: \(qc'.view < qc^*.view\). The formation of \(qc'\) arises from \(2f+1\) nodes voting for \(B'\) and \(qc\), indicating that at least \(f+1\) honest nodes have seen \(qc\). In any view higher than \(qc'.view\), the primary node collects \(2f+1\) new-view messages, at least one of which will include \(qc\) or a QC inherited from \(qc\) (Lemma~\ref{lemma:unique-extend}). Since \(qc^*\) does not inherit from \(qc\), the quorum certificate \(qc^*\) for the conflicting block \(B^*\) cannot form.

    In conclusion, \(B^*\) cannot be committed by any honest node, rendering the assumption invalid. Therefore, it is impossible for two conflicting blocks \(B\) and \(B^*\) to be committed by different honest nodes respectively.
    
\end{proof}

Lemma~\ref{lemma:non-conflicating block},~\ref{lemma:unique-extend} and~\ref{lemma:unique-commit} ensure that all honest nodes commit a consistent block. However, since our protocol organizes transactions in microblocks and the block contains only microblock identifiers, it is also necessary to confirm that the set of requests referenced by each block is identical. This requires verification that all nodes decode the same transactions from these microblock identifiers to achieve a fully consistent state across the network.

\begin{lemma}
    \label{lemma:Block consistency}
    If two honest nodes have the same block $B$, then it is guaranteed that they can construct the same transaction list.
\end{lemma}

\begin{proof}
    Since all microblocks referenced in the block are accompanied by an AC, Lemma~\ref{lemma:Microblock consistency} ensure that each microblock associated with an AC can be consistently retrieved. Moreover, Theorem~\ref{Theorem:SMP chain-consistency} guarantees that the predecessors of these microblocks are also accessible and consistent. Therefore, all honest nodes can obtain a consistent set of microblocks. Utilizing a deterministic local sorting algorithm, each node can achieve a consistent ordering of these requests.
\end{proof}

Combing the above lemmas, the proof of Theorem~\ref{theorem:safety} proceeds as follows.
\begin{proof}
    Lemma~\ref{lemma:unique-commit} ensures that all honest nodes will commit the same block at the same blockchain height, while Lemma~\ref{lemma:Block consistency} guarantees that all honest nodes can derive a consistent transaction list based on a consistent block. Therefore, all honest nodes will execute requests in a consistent order.
\end{proof}


%


\end{document}